\documentclass{sig-alternate-10}
\usepackage{verbatim}
\usepackage{graphicx}
\begin{document}



\conferenceinfo{}{}

\newtheorem{thm}{Theorem}
\newtheorem{theorem}{Theorem}
\newtheorem{thm1}{Theorem}
\newtheorem{rem}[thm1]{Remark}
\newtheorem{thm2}{Theorem}
\newtheorem{defi}[thm2]{Definition}
\newtheorem{fact}[thm]{Fact}
\newtheorem{note}[thm]{Note}
\newtheorem{notation}[thm]{Notation}
\newtheorem{definition}[thm]{Definition}
\newtheorem{proposition}[thm]{Proposition}
\newtheorem{lem}[thm]{Lemma}

\newtheorem{cor}[thm]{Corollary}

\newcommand{\E}{\mathbb{E}}
\newcommand{\NN}{\mathrm{NN}}
\newcommand{\IP}{\mathrm{IP}}
\newcommand{\BCS}{\mathrm{BCS}}
\newcommand{\RCS}{\mathrm{RCS}}
\newcommand{\JS}{\mathrm{JS}}
\newcommand{\tb}{\mathrm{\psi}}
\newcommand{\tu}{\mathrm{{\psi}_u}}
\newcommand{\tm}{\mathrm{\Psi}}
\newcommand{\dH}{\mathrm{d_H}}
\newcommand{\N}{\mathrm{N}}
\newcommand{\Ham}{\mathrm{Ham}}
\newcommand{\Sim}{\mathrm{Sim}}
\newcommand{\D}{\mathcal{D}}
\newcommand{\Var}{\mathrm{Var}}
\newcommand{\Cov}{\mathrm{Cov}}
\newcommand{\G}{\mathcal{G}}
\newcommand{\R}{\mathbb{R}}
\renewcommand{\O}{\tilde{O}}
\newcommand\numberthis{\addtocounter{equation}{1}\tag{\theequation}}


\title{Similarity preserving compressions of high dimensional sparse data}
%
%
%
%
%

\numberofauthors{2} 
%
\author{
\alignauthor
Raghav Kulkarni\\
       \affaddr{LinkedIn Bangalore}\\
       \email{{kulraghav@gmail.com}}
\alignauthor
Rameshwar Pratap\\
       \affaddr{IIIT Bangalore}\\
       \email{rameshwar.pratap@gmail.com}
}

\maketitle
\begin{abstract}
The rise of internet has resulted in an explosion of data consisting of 
millions of articles, images, songs, and videos. Most of this data is 
high dimensional and sparse. The need to perform an efficient search 
for similar objects in such high dimensional big datasets is becoming 
increasingly common. Even with the rapid growth in computing power, 
the brute-force search for such a task is impractical and at times 
impossible. Therefore algorithmic solutions such as Locality Sensitive Hashing (LSH) 
are required to achieve the desired efficiency in search.

Any similarity search method that achieves the efficiency uses one (or both) 
of the following methods: $1.$ Compress the data by reducing its dimension 
while preserving the similarities between any pair of data-objects $2.$ 
Limit the search space by grouping the data-objects based on their similarities.
Typically $2$ is obtained as a consequence of $1$.

Our focus is on high dimensional sparse data, where the standard compression schemes, 
such as LSH for Hamming distance (Gionis, Indyk and Motwani~\cite{GIM99}), 
become inefficient in both $1$ and $2$ due to at least one of the following reasons:  
$1$. No efficient compression schemes mapping binary vectors to binary vectors
$2$. Compression length is nearly linear in the dimension and grows inversely with the sparsity
$3$. Randomness used grows linearly with the product of dimension and compression length.

We propose an efficient compression scheme mapping binary vectors into binary vectors 
and simultaneously preserving Hamming distance and Inner Product. Our schemes avoid all 
the above mentioned drawbacks for high dimensional sparse data. The length of our 
compression depends only on the sparsity and is independent of the dimension of the data. 
Moreover our schemes provide one-shot solution for Hamming distance and Inner Product, and 
work in the streaming setting as well. In contrast with the ``local projection'' strategies 
used by most of the previous schemes, our scheme combines (using sparsity) the following 
two strategies: $1.$ Partitioning the dimensions into several buckets, $2.$ Then obtaining 
``global linear summaries'' in each of these buckets. We generalize our scheme for real-valued 
data and obtain compressions for Euclidean distance, Inner Product, and $k$-way Inner Product.
\end{abstract}

\section{Introduction}
The technological advancements have led to the generation of huge
amount of data over the web such as texts, images, audios, and videos. 
Needless to say that most of these datasets are high dimensional. 
Searching for similar data-objects in such massive and high dimensional datasets is becoming a fundamental 
subroutine in many scenarios like clustering, classification, nearest neighbors, ranking etc. 
However, due to the \textit{``curse of dimensionality''} 
a brute-force way to compute the similarity scores   on such data sets is very expensive and at times infeasible.
Therefore it is quite natural to investigate the techniques that compress the dimension of dataset 
while  preserving the similarity between data objects. There are various compressing schemes 
that have been already studied for different similarity measures. 
We would like to emphasize  that any such compressing scheme is useful only when it satisfies the following guarantee, 
\textit{i.e.} when data objects are ``nearby'' (under the desired similarity measure), then they should 
remain near-by in the compressed version, and when they are ``far'', they should remain far in the
compressed version. In the case of probabilistic compression schemes the above should happen with high probability. 
Below we discuss a few such notable schemes. In this work we consider  binary and real-valued datasets.
For binary data we focus on Hamming distance and Inner product, while for real-valued data we focus on 
Euclidean distance and Inner product.

\subsection{Examples of similarity preserving compressions }
Data objects in a datasets can be considered as points (vectors) in high dimensional space. 
 Let we have  $n$ vectors (binary 
or real-valued) in $d$-dimensional space.  

\begin{itemize}
\item Gionis, Indyk, Motwani~\cite{GIM99} proposed a data structure to solve  
approximate nearest neighbor ($c$-$\NN$) problem 
in binary data for {\bf Hamming distance}.  Their scheme popularly  known 
as {\bf Locality Sensitive Hashing} (LSH). Intuitively, their data structure can be 
viewed as a compression of a binary vector, which is obtained by projecting it on a 
randomly chosen bit positions. 
\item \textbf{JL transform}~\cite{JL83} suggests a compressing scheme for real-valued data.   
For any $\epsilon>0$, it compresses the dimension of the points from $d$ to 
$O\left(\frac{1}{\epsilon^2}\log n\right)$ while preserving 
the \textbf{Euclidean distance} between any pair of points within factor of $(1\pm \epsilon)$.

\item Given two vectors $\mathbf{u},\mathbf{v} \in \R^d$, the \textbf{inner product similarity}
between them is defined as $\langle \mathbf{u},\mathbf{v}\rangle :=\Sigma_{i=1}^d\mathbf{u}[i] \mathbf{v}[i].$  
Ata Kab\'{a}n~\cite{Kaban15} suggested a compression schemes for real data which preserves   inner product 
\textit{via} \textbf{random projection}. On the contrary, if the input data is binary, and it  
is desirable to get the compression only in binary data, then to the best of our 
knowledge no such compression scheme is available which achieves a non-trivial compression. 
However, with some sparsity assumption (bound on the number of $1$'s),  
there are some schemes available  which \textit{via} asymmetric  padding (adding a few extra bits in the vector) 
reduce the inner product similarity (of the original data) to the Hamming~\cite{BeraP16}, 
and Jaccard similarity (see Preliminaries for a definition)~\cite{ShrivastavaWWW015}. Then the compression scheme for 
Hamming or Jaccard can be applied on the padded version of the data.

\item Binary data can also be viewed as a collection of sets, then the underlying similarity 
measure of interest  can be the \textbf{Jaccard similarity}. 
Broder \textit{et. al.}~\cite{Broder00,BroderCFM98,BroderCPM00} suggested a compression scheme for preserving Jaccard similarity between sets  
which is popularly known as \textbf{Minwise permutations}.

\end{itemize}
\subsection{Our focus: High dimensional (sparse) data}
In this work, we focus  on High Dimensional Sparse Data.  In many real-life scenarios, 
data object is represented as very high-dimensional but sparse vectors, \textit{i.e.} 
number of  all possible attributes (features) is huge, however, each data object has 
only a very small subset of attributes. For example, in bag-of-word representation of 
text data, the number of dimensions equals to the size of vocabulary, which is large. 
However for each data point, say a document, contains only a small number of words in 
the vocabulary, leading to a sparse vector representation. The bag-of-words representation 
is also commonly used for image data. Data-sparsity is commonly prevalent in audio and 
video-data as well.

\subsection{Shortcomings of earlier schemes for high dimensional (sparse) data}
The quality of any compression scheme can be evaluated based on the following two parameters - 1) 
the \textit{compression-length}, and  2) the amount of \textit{randomness} required for the compression. 
The compression-length is defined as the dimension of the data after compression.
  Ideally, it is desirable to
 have  both of these to be small while preserving a desired accuracy in the compression. Below
  we will notice that  most of the above mentioned compression schemes   
 become in-feasible in the case of high dimensional  sparse datasets as 1) their 
  compression-length is  very high, and 
 2) the amount of randomness required for the 
 compression is quite huge. 

\begin{itemize}
 \item \textbf{Hamming distance:} Consider the problem of finding $c$-$\NN$ (see Definition~\ref{definition:cNN})
 for Hamming distance in binary data.  In the  LHS scheme, the size of hashtable determines the compression-length.  
The size of hashtable $K=O\left(\log_ {\frac{1}{p_2}} n \right)$ (see Definition~\ref{definition:LSH}). 
If $r=O(1)$, then the size of hashtable $K=O\left(\log_ {\frac{1}{p_2}} n \right)= O(\frac{d}{cr}\log n)=O(d\log n)$, 
which is linear in the dimension. Further,   in order to randomly choose a 
 bit position (between $1$ to $d$), it is require to generate $O(\log d)$ many random bits. 
    Moreover, as the size of hash table is $K$,  
    and the number of hash tables is $L$, it is required to generate $O(KL\log d)$ many random bits to create the hashtable, 
    which become quite large specially when $K$ is linear in $d$.

  \item \textbf{Euclidean distance:} In order to achieve compression that preserve the distance 
  between any pair of points, due to JL transform~\cite{JL83, Achlioptas03}, it is required 
  to project the input matrix on a random matrix of 
  dimensions $d\times k$, where $k=O\left(\frac{1}{\epsilon^2}\log n\right)$. Each entry of the
  random matrix is chosen from $\{\pm 1\}$  with probability  $\frac{1}{2}$ (see~\cite{Achlioptas03}), 
  or from a normal distribution (see~\cite{JL83}).
  The compression-length in this scheme is  $O\left(\frac{1}{\epsilon^2}\log n\right) $, and it requires \\$O\left(\frac{1}{\epsilon^2}d\log n\right)$
    randomness. 
  
  \item \textbf{Inner product:} 
  Compression schemes which compress binary data into binary data while preserving Inner 
  product is not known. However using \textit{``asymmetric padding scheme''} of~\cite{BeraP16,ShrivastavaWWW015} 
  it is possible to get a compression via Hamming or Jaccard Similarity measure, then shortcomings of Jaccard and Hamming will 
  get carry forward in such scheme. Further, in  case of   real valued data the compression scheme of 
  Ata Kab\'{a}n~\cite{Kaban15}  has  
  compression-length  $=O\left(\frac{1}{\epsilon^2}\log n\right) $, and requires $O\left(\frac{1}{\epsilon^2}d\log n\right)$
    randomness. 

    \item\textbf{Jaccard Similarity:} Minhash permutations~\cite{Broder00,BroderCFM98,BroderCPM00} suggest a compression scheme for preserving Jaccard similarity 
  for   a collection of sets. A major disadvantage of this scheme is that for high dimensional 
  data computing permutations are very expensive, and further in order to achieve a reasonable 
  accuracy in the compression a larger number of repetition might be required. 
  A major disadvantage of this scheme is that it requires substantially large 
  amount of randomness that grows polynomially  in the dimension. 
   \end{itemize}
  \vspace{-0.5cm}
   \paragraph{Lack of good binary to binary compression schemes}
  To summarize the above, there are two main compression schemes currently available for binary to binary compression. 
  The first one is LSH and the second one is JL-transform.
  The LSH requires the compression size to be linear in the dimension and the JL-transform can achieve logarithmic 
  compression size but it will compress binary vectors to real vectors. The analogue of JL-transform which compresses 
  binary vectors to binary vectors requires the compression-length to be linear in the number of data points
  (see Lemma~\ref{lem:analogousJL}.)
  Since both dimension as well as the number of data points can be large, these schemes are inefficient.
  In this paper we propose an efficient binary to binary compression scheme for sparse data which works simultaneously 
  for both Hamming distance and Inner Product.
   \subsection{Our contribution}
In this work we present a compressing scheme for high dimensional sparse  data. 
In contrast with the ``local projection''
strategies used by most of the previous schemes such as LSH~\cite{IM98,GIM99} and JL~\cite{JL83}, 
our scheme combines (using sparsity) the following two step approach 1. Partitioning the dimensions 
into several buckets, 2. Then obtaining ``global linear summaries'' of each of these buckets.
 We present our result below:
\subsubsection{For binary data}
For binary data, our compression scheme  provides  one-shot solution for  both Hamming and Inner 
product --  compressed data preserves both Hamming distance and Inner product.
Moreover, the compression-length depends only on the sparsity of data and is independent of the dimension of data. 
We first informally state our compression scheme for binary data, see Definition~\ref{defi:bcs} 
for a formal definition.

Given a binary vector $\textbf{u}\in \{0,1\}^{d}$, our scheme compress it into a 
$\N$-dimensional binary vector (say) $\mathbf{u'}\in\{0,1\}^{\N}$ as follows, where $\N$ to be specified later. 
We randomly map each bit position (say) $\{i\}_{i=1}^d$ of the original
data to an integer $\{j\}_{j=1}^{\N}$. To compute the $j$-th bit of the compressed vector $\mathbf{u'}$ 
we check which bits positions have been mapped to $j$,  we compute 
the parity of bits located at those positions, and assign it to $\mathbf{u'}[j].$ 
The following figure illustrate an example of the compression. 
 \begin{figure}[ht!]
\centering
\includegraphics[scale=.033]{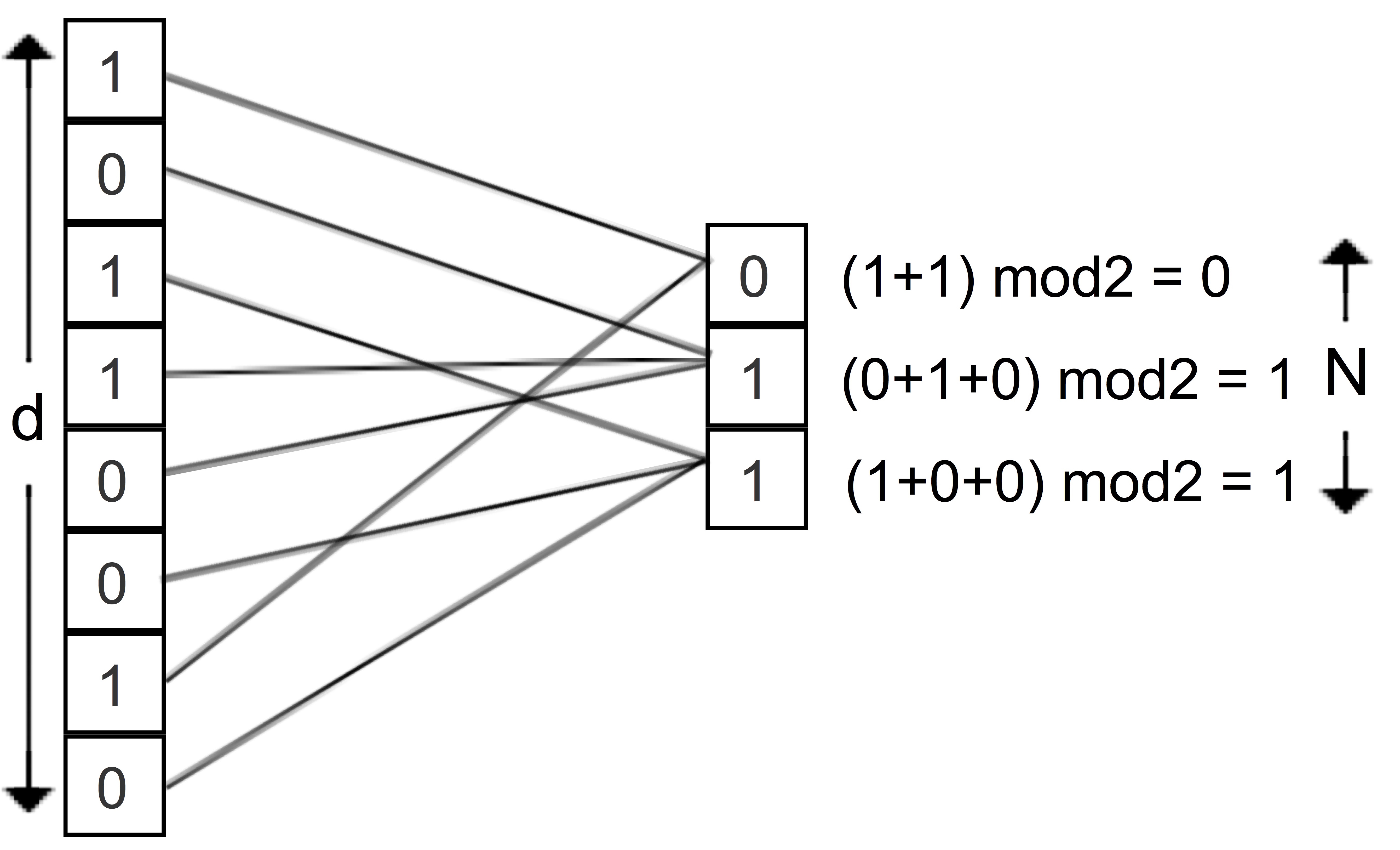}
\end{figure}
In the following theorems  let $\tb$ denote the maximum number of $1$ in any vector.
We state our result for binary data as follows:
\begin{theorem}\label{theorem:compressionHamming}
Consider a set $\mathrm{U}$ of binary vectors \\$\{\mathbf{u_i}\}_{i=1}^n\subseteq \{0, 1\}^d$, 
  a positive  integer $r$, and   $\epsilon>0$. 
  If $\epsilon r >3 \log n$,  we set $\N=O({\tb}^2)$;  if $\epsilon r < 3 \log n$, 
    we set $\N=O({\tb}^2\log^2n) $,  and compress them 
  into  a set  $\mathrm{U'}$ of binary vectors $\{\mathbf{u_i'}\}_{i=1}^n\subseteq\{0, 1\}^{\N}$  using 
  our Binary Compression  Scheme. 
    Then for all  $\mathbf{u_i}, \mathbf{u_j}\in \mathrm{U}$, 
\begin{itemize}
 \item if $\dH(\mathbf{u_i}, \mathbf{u_j})< r$, then $\Pr [\dH({\mathbf{u_i}}', {\mathbf{u_j}}')< r]=1$,
  \item if $\dH(\mathbf{u_i}, \mathbf{u_j})\geq (1+\epsilon)r$, then $\Pr [\dH({\mathbf{u_i}}', {\mathbf{u_j}}')< r]<\frac{1}{n}.$
\end{itemize}
\end{theorem}
\begin{theorem}\label{theorem:compressionIP}
Consider a set $\mathrm{U}$ of binary vectors \\$\{\mathbf{u_i}\}_{i=1}^n\subseteq \{0, 1\}^d$, 
 a positive  integer $r$, and   $\epsilon>0$. 
  If $\epsilon r >3 \log n$,  we set $\N=O({\tb}^2)$;  if $\epsilon r < 3 \log n$,  
    we set $\N=O({\tb}^2\log^2n) $,  and compress them into 
    a set  $\mathrm{U'}$ of binary vectors
    $\{\mathbf{u_i'}\}_{i=1}^n\subseteq\{0, 1\}^{\N}$  using our Binary Compression   Scheme.   
    Then for all $\mathbf{u_i}, \mathbf{u_j}\in \mathrm{U}$ the following is true with probability
 at least $1-\frac{1}{n}$,
 \[
  (1-\epsilon)\IP(\mathbf{u_i}, \mathbf{u_j})\leq \IP({\mathbf{u_i}}', {\mathbf{u_j}}')\leq (1+\epsilon)\IP(\mathbf{u_i}, \mathbf{u_j}).
 \]
\end{theorem}
In the following theorem, we strengthen our result of Theorem~\ref{theorem:compressionHamming}, 
and shows a compression bound which is independent of the  dimension 
and the sparsity, but depends only on the Hamming distance between 
the vectors. However, we could show our result in the Expectation, and only for a pair of vectors.
\begin{theorem}\label{theorem:compressionR}
Consider two binary vectors $\mathbf{\mathbf{u}}, \mathbf{v} \\ \in\{0, 1\}^d$, which get compressed into  
 vectors $\mathbf{\mathbf{u'}}, \mathbf{v'} \in \{0, 1\}^{\N}$  using our Binary Compression Scheme.
  If we set  $\N=O(r^2)$, then
 \begin{itemize}
  \item if $\dH(\mathbf{u}, \mathbf{v})< r$, then $\Pr [\dH({\mathbf{u}}', {\mathbf{v}}')< r]=1$, and 
  \item if $\dH(\mathbf{\mathbf{u}}, \mathbf{v})\geq 4r$,  then $\E[\dH(\mathbf{\mathbf{u'}}, \mathbf{v'})]>2r.$
 \end{itemize}
 \end{theorem}
\begin{rem}
 To the best of our knowledge, ours is the first efficient binary to binary 
 compression scheme for preserving Hamming distance and Inner product. For 
 Hamming distance in fact our scheme obtains the ``no-false-negative''
 guarantee analogous to the one obtained in recent paper by Pagh~\cite{Pagh16}.
\end{rem}
\begin{rem}
 When $r$ is constant, as mentioned above, 
 LSH~\cite{GIM99} requires compression 
 length linear in the dimension. However, due to Theorem~\ref{theorem:compressionR}, our compression length 
 is only constant.
\end{rem}
\begin{rem}
 Our compression length is $O(\tb \log^2n)$, which is independent of the dimension $d$; 
 whereas other schemes such as LSH may require the compression length growing linearly in $d$ 
 and the analogue of JL-transform for binary to binary compression requires compression 
 length growing linearly in $n$ (see Lemma~\ref{lem:analogousJL}). 
\end{rem}
\begin{rem} The randomness used by our compression  scheme is $O(d \log \N)$ 
which grows logarithmically in the compression length $N$ whereas  the 
JL-transform uses randomness growing linearly in the compression length. 
For all-pair compression for $n$ data points we use $O(d (\log \tb + \log \log n))$ randomness, 
which grows logarithmically in the sparsity and sub-logarithmically in terms of number of data points. 
\end{rem}
\vspace{-0.2cm}
\subsubsection{For real-valued data}
We generalize our scheme for real-valued data also and obtain compressions for Euclidean distance, 
Inner product, and $k$-way Inner product. We first state our compression scheme as follows:

Given a vector $\textbf{a}\in \R^{d}$, our scheme compress it into a 
$\N$-dimensional vector (say) $\boldsymbol{\alpha}^{\N}$ as follows. 
We randomly map each coordinate position (say) $\{i\}_{i=1}^d$ of the original
data to an integer $\{j\}_{j=1}^{\N}$. To compute the $j$-th coordinate of the 
compressed vector  $\boldsymbol{\alpha}$ we check which coordinates of the original data  have been 
mapped to $j$, we multiply  the numbers located at those positions with a random variable $x_i$, 
 compute their summation, and assign it to  $\boldsymbol{\alpha}[j]$, 
 where $x_i$  takes a value between $\{-1, +1\}$ with probability $1/2.$
 The following figure illustrate an example of the compression. 
  \begin{figure}[ht!]
\centering
\includegraphics[scale=.033]{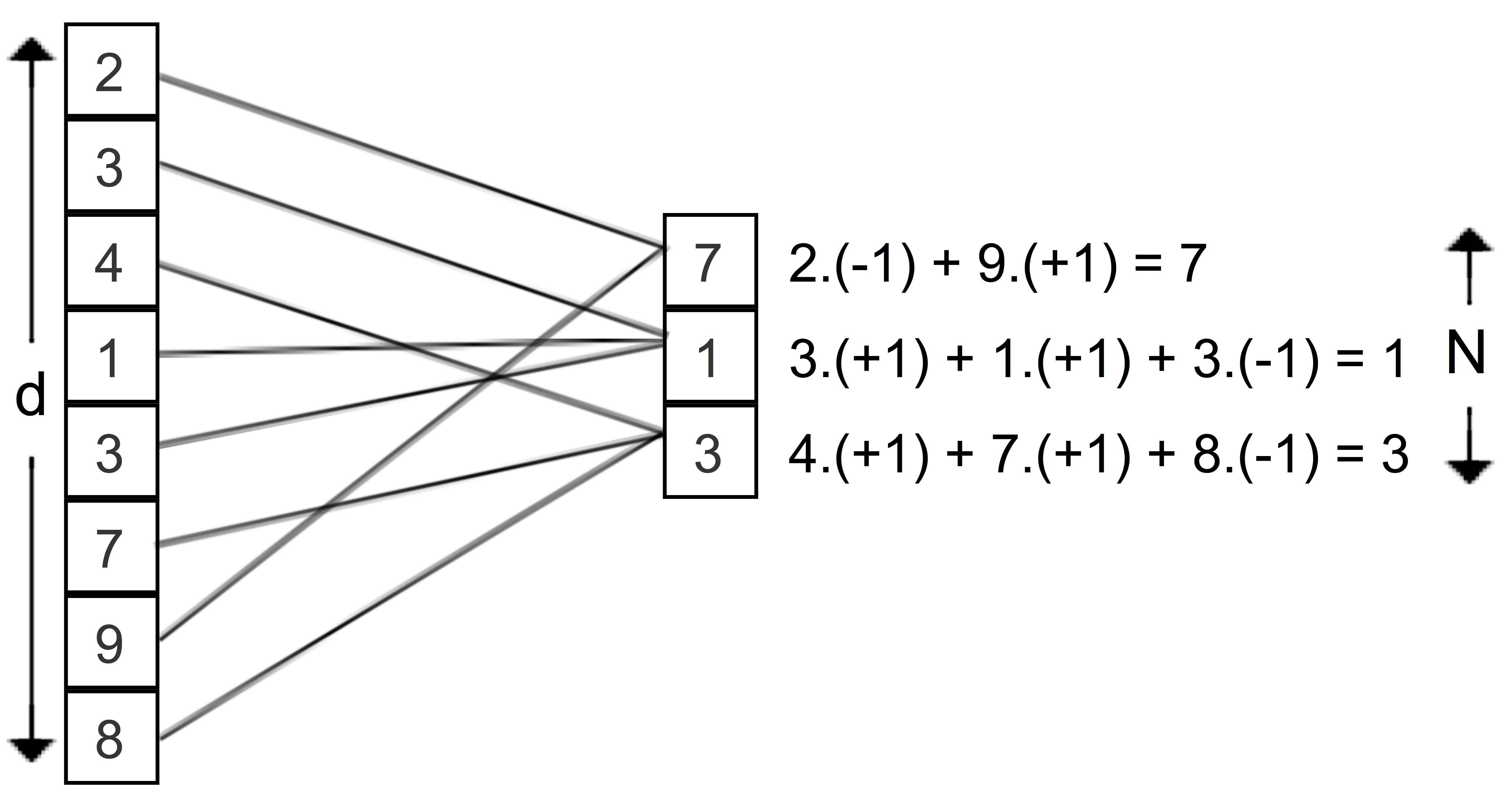}
\end{figure}
In the following we present our main result for real valued data which is compression bound for preserving 
$k$-way inner product. For a set of $k$ vectors $\{\boldsymbol{\alpha_i}\}_{i=1}^k\in \R^d$, their 
$k$-way inner product is defined as 
$$\langle \boldsymbol{\alpha}_1\boldsymbol{\alpha}_2\ldots\boldsymbol{\alpha}_k\rangle
=\sum_{j=1}^d\boldsymbol{\alpha}_1[j]\boldsymbol{\alpha}_2[j]\ldots\boldsymbol{\alpha}_k[j],$$
where $\boldsymbol{\alpha}_1[j]$ denote the $j$-th coordinate of the vector $\boldsymbol{\alpha}_1$.
\begin{theorem}\label{theorem:compressionRealKway}
 Consider a set of $k$ vectors $\{\mathbf{a_i}\}_{i=1}^k\in \R^d$, 
  which get compressed into vectors $\{\boldsymbol{\alpha_i}\}_{i=1}^k \in \R^{\N}$   using our Real Compression Scheme. 
   If we set $\N=\frac{10\tm^k}{\epsilon^2}$, 
 where $\tm=\max\{||{\mathbf{a_i}||^2}\}_{i=1}^k$ and $\epsilon>0$, then the following holds 
 \[
 \Pr\left[\left|\langle\boldsymbol{\alpha}_1\boldsymbol{\alpha}_2\ldots\boldsymbol{\alpha}_k\rangle- \langle\mathbf{a}_1\mathbf{a}_2\ldots\mathbf{a}_k\rangle \right|>\epsilon \right]<1/10.
 \]
\end{theorem}
\begin{rem}
 An advantage of our compression scheme is that it can be constructed in the
 streaming model quite efficiently. The only requirement is that in the  case of 
 binary data the maximum number of $1$ the  vectors in the stream should be bounded, 
 and in the case of real valued data norm of the vectors should be bounded. 
\end{rem}
\vspace{-0.2cm}
\subsection{Comparison with previous work}
A major advantage of our compression scheme is that it provides a one-shot solution for different similarity 
measures -- Binary compression scheme preserves both Hamming distance and Inner product, 
and Real valued data compression scheme preserves both Euclidean distance,  Inner product, and $k$-way Inner product. 
The second main advantage of our compression scheme for binary data   it gives a binary to binary 
compression as opposed to the binary to real compression by JL-transform. 
Third main advantage is that our compression scheme is that its compression 
size is independent of the dimensions and depends only on the sparsity as opposed to 
Gionis, Indyk, Motwani~\cite{GIM99} scheme which requires linear size compression. For real-valued 
data our results are weaker compared to previous known works but they generalize to $k$-way 
inner product, which none of the previous work does.  
Another advantage of our real valued compression scheme is that 
when the number of points are small (constant),  then for preserving 
a pairwise Inner Product  or Euclidean distance, 
we have a clear advantage on  the amount of  randomness required for the 
compression, 
the randomness required by our scheme 
 grows logarithmically in the compression length, whereas  the 
other schemes require  randomness which  grows linearly in the compression length.
\subsection*{Potential applications}
A potential use  of our result is to improve approximate nearest neighbor search via composing with LSH. 
Due to the ``curse of dimensionality'' many search algorithms scale poorly in high dimensional data. 
So, if it is possible to get a succinct compression of data while preserving the similarity 
score between pair of data points, then such compression naturally helps for efficient  search.  
One can first compress the input such that it preserve the desired 
similarity measure, and then  can apply a collision based hashing algorithm such as 
LSH~\cite{GIM99, IM98} for efficient approximate nearest neighbor ($c$-$\NN$)  on the compressed data. 
As our compression scheme provides a similar guarantee as of 
 Definition~\ref{definition:LSH}, then 
 one can construct data structure for LSH 
 for approximate nearest neighbor problem.
Thus, our similarity preserving compression scheme leads to an efficient approximate nearest neighbor search. 

There are many similarity based algorithmic methods used in large scale learning and information retrieval, 
e.g., Frequent itemset mining~\cite{AgrawalS94}, ROCK clustering \cite{ROCK}. One could potentially obtain algorithmic 
speed up in these methods via our compression schemes. 
Recently compression based on LSH for inner-product is used to speed up the  forward and back-propagation 
in neural networks \cite{deeplearning}. One could potentially use our scheme to take advantage of sparsity and 
obtain further speed up.
\vspace{-0.4cm}
\subsection*{Organization of the paper}
In Section~\ref{sec:Background}, we present the necessary background which 
helps to understand the paper. In Section~\ref{sec:BinaryResult}, we present 
our compression scheme for high dimensional sparse binary data. 
 In Section~\ref{sec:RealResult}, we present our compression scheme for high 
 dimensional sparse real data.  Finally in Section~\ref{sec:Conclusion}, we 
 conclude our discussion, and state some possible extensions of the work.

\section{Background}\label{sec:Background}
\begin{tabular}{|c|l|}
\hline
 \multicolumn{2}{|c|}{\bf Notations}\\
 \hline
 $\N$ & number of coordinates/bit positions\\ &in the compressed data \\
 \hline
$\tb$ & upper bound on the number of $1$'s in \\&any binary vector.\\
\hline
$\tm$ & upper bound on the norm of any \\&real-valued vector.\\
\hline
$||\mathbf{a}||$ & $l_2$ norm of the vector $\mathbf{a}$\\
\hline
$\mathbf{a}[i]$ & $i$-th bit position (coordinate) of  \\&binary (real-valued) vector $\mathbf{a}$ .\\
\hline
 $\dH(\mathbf{u}, \mathbf{v})$& Hamming distance between binary\\& vectors $\mathbf{u}$ and $\mathbf{v}.$\\
 \hline
$\IP(\mathbf{a}, \mathbf{b})$ & Inner product between binary/\\&real-valued vectors $\mathbf{a}$ and $\mathbf{b}.$\\
\hline
 \end{tabular}
 \subsection{Probability background}
 \begin{definition}\label{definition:varDef}
  The Variance of a random variable $X$, denoted $\Var(X)$, is defined as the expected value of the squared deviation 
  of $X$ from its mean. 
  \[
   \Var(X)=\E[(X-\E(X))^2]=\E(X^2)-\E(X)^2.
  \]
 \end{definition}
 
 \begin{definition}\label{definition:coVarDef}
 Let $X$ and $Y$ be jointly distributed random variables. The \textit{Covariance} of $X$
and $Y$, denoted $\Cov(X, Y)$, is defined as
  \[
   \Cov(X, Y)=\E[(X-\E(X))(Y-\E(Y))]. 
  \]
 \end{definition}

 \begin{fact}\label{fact:varProp}
  Let $X$ be a random variable and $\lambda$ be a constant. Then,
$\Var(\lambda+X)=\Var(X)$ and $\Var(\lambda X)=\lambda^2\Var(X).$
 \end{fact}
 
 \begin{fact}\label{fact:varProp1}
  Let $X_1, X_2, \ldots, X_n$ be a set of $n$ random variables. Then, 
\[
 \Var\left(\sum_{i=1}^n X_i \right)=\sum_{i=1}^n\Var\left( X_i \right)+\sum_{i \neq j} \Cov(X_i, X_j).
\]
\end{fact}
\begin{fact}\label{fact:coVarProp}
  Let $X$ and $Y$ be a pair of random variables and $\lambda$ be a constant. Then,
$\Cov(\lambda X, \lambda Y)=\lambda^2\Cov(X, Y).$
 \end{fact}
\begin{fact}[Chebyshev's inequality]\label{fact:Chebyshev}
  Let $X$  be a  random variable having finite mean and finite non-zero variance $\sigma^2$. Then for any real number
  $\lambda>0,$
  \[
   \Pr[|X-\E(X)|\geq\lambda \sigma] \leq \frac{1}{\lambda^2}.
  \]
 \end{fact}

 \subsection{Similarity measures and their respective compression schemes}
\paragraph{Hamming distance} Let $\mathbf{u}, \mathbf{v}\in \{0, 1\}^d$ be two binary vectors, then  the 
Hamming distance between these two vectors is the number of bit positions where they differ. 
 To the best of 
our knowledge, there does not exist any non-trivial compression scheme which provide  
similar compression guarantees such as JL-lemma provides for Euclidean distance. 
In the following lemma, we show that for a set of $n$-binary vectors an analogous JL-type binary to binary 
compression (if it exist) may require compression length linear in $n$. 
Further collision \footnote{A collision occurs when two object hash to the same hash value.} based hashing scheme  such as LSH (due to  
Gionis \textit{et al.}~\cite{GIM99}, see Subsection~\ref{subsection:subsecLSH}) can be considered as a  
binary to binary compression scheme, where the size of hashtable determines the compression-length. Their techniques includes randomly choosing bit positions and checking if the query and input vectors are matching exactly at those bit positions.
\begin{lem}\label{lem:analogousJL}
 Consider a set of $n$-binary vectors, then  an analogous JL-type binary to binary compression (if it exist) 
 may require compression length linear in $n$.
\end{lem}
\begin{proof}
  Consider a set of  $n$ binary vectors $\{e_i\}_{i=1}^n$ -- standard unit vectors, and the zero vector $e_0$. 
  The Hamming distance between $e_0$ and any $e_i$ is $1$, and the Hamming distance between any pair of vectors 
  $e_i$ and $e_j$ for $i \neq j$ is $2$. 
  Let   $f$ be a map which map these points into binary vectors of dimension $k$ by preserving the distance between 
  any pair of vectors within  a factor of $1\pm \varepsilon$, for a parameter $\varepsilon>0$. Thus, these $n$ points $\{f(e_i)\}_{i=1}^n$ are within a
  distance at most $(1+\varepsilon)$ from $f(e_0)$, and any two points $f(e_i)$ and $f(e_j)$ for $i \neq j$ are at distance 
  at least $2(1-\varepsilon)$. However, the total number of  points at distance at most $(1+\varepsilon)$ from  
  $f(e_0)$ is $O(k^{1+\varepsilon})$, and distance between any  two points $f(e_i)$ and $f(e_j)$ for $i \neq j$ 
  is non-zero so each point $\{e_i\}_{i=1}^n$ has its distinct image.  Thus 
  $O(k^{1+\varepsilon})$ should be equal to $n$, 
  which gives $k=\Omega(n^{\frac{1}{1+\varepsilon}})$. Thus the compression
  length can  be linear in $n$. 
\end{proof}
\paragraph{Euclidean distance} 
Given two vectors $\mathbf{a}, \mathbf{b} \in \R^d$, the Euclidean distance 
between them is denoted as $||\mathbf{a}, \mathbf{b}||$ and defined as 
\\ $\sqrt{\Sigma_{i=1}^d(\mathbf{a}[i]- \mathbf{b}[i])^2}.$ 
A classical result by Johnson and Lindenstrauss~\cite{JL83} suggest 
a compressing scheme  which 
for any set $\D$ of $n$ vectors in $\R^d$ preserve 
pairwise Euclidean distance between any pair of vectors in $\D$.  
\begin{lem}[JL transform~\cite{JL83}]
 For any $\epsilon\in (0, 1)$,  and any integer $n$, 
 let $k$ be a positive integer such that $k=O\left( \frac{1}{\epsilon^2}\log n \right)$. 
 Then for any set $\D$ of $n$ vectors in $\R^d$, there is a map $f: \R^d\rightarrow \R^k$ 
 such that for any pair of vectors $\mathbf{a}, \mathbf{b}$ in $\D:$
 \[
  (1-\epsilon)||\mathbf{a}, \mathbf{b}||^2\leq ||f(\mathbf{a}), f(\mathbf{b})||^2\leq  (1+\epsilon)||\mathbf{a}, \mathbf{b}||^2
 \]
Furthermore, the mapping $f$ can be found in randomized polynomial time.
\end{lem}
In several followup works on JL lemma, the function $f$ has been regarded as a random 
projection matrix $R\in \R^{d\times k}$, and can be constructed element-wise using 
Gaussian due to Indyk and Motwani~\cite{IM98}, 
or uniform $\{+1,-1\}$ due to Achlioptas~\cite{Achlioptas03}.

\paragraph{Inner product}
Given two vectors $\mathbf{u}, \mathbf{v} \in \R^d$, the Inner product  $\langle \mathbf{u}, \mathbf{v}\rangle $ 
between them is defined as $$\langle \mathbf{u}, \mathbf{v}\rangle :=\Sigma_{i=1}^d\mathbf{u}[i]\mathbf{v}[i].$$  
Compression schemes which preserves Inner product  has been studied quite a lot in the recent time.
In the case of binary data, along with some sparsity assumption (bound on the number of $1$'s), 
there are some schemes available  which by padding (add a few extra bits in the vector)  
reduce the Inner product  (of the original data) to the Hamming~\cite{BeraP16}, 
and Jaccard similarity~\cite{ShrivastavaWWW015}. Then the compression scheme for 
Hamming or Jaccard can be applied on the padded version of the data. Similarly, 
in the case of real-valued data, a similar padding technique is known that due padding 
reduces  Inner product  to 
Euclidean distance~\cite{Shrivastava014}.  Recently, an interesting work by
Ata Kab\'{a}n~\cite{Kaban15} suggested a compression schemes \textit{via} random projection method. Their scheme
approximately preserve Inner Product between any pair of input points and their compression bound 
matches  the bound of JL-transform~\cite{JL83}.

\paragraph{Jaccard similarity} Binary vectors can also be considered as sets over 
the universe of all possible features, and a set contain only those elements which 
have non-zero entries in the corresponding binary vector. For example two vectors $\mathbf{u}, \mathbf{v}\in \{0, 1\}^d$ 
can  be viewed as two sets $\mathbf{u}, \mathbf{v}\subseteq \{1, 2, \ldots d\}$. 
Here, the underlying similarity measure of interest  is the Jaccard 
similarity which is defined as follows 
$\JS(\mathbf{u}, \mathbf{v})=\frac{|\mathbf{u} \cap \mathbf{v}|}{|\mathbf{u} \cup \mathbf{v}|}.$

A celebrated work by Broder 
\textit{et al.}~\cite{Broder00,BroderCFM98,BroderCPM00} suggested a 
technique to compress a collection of sets while preserving the  Jaccard similarity between any pair of sets.
 Their technique includes taking a random permutation 
of $\{1, 2, \ldots, d\}$ and assigning a value to each set which maps to  
minimum  under that permutation. This compression scheme is popularly known as Minwise hashing. 
\begin{definition}[Minwise Hash function]\label{defn:minwise}
    Let $\pi$ be a permutations over $\{1, \ldots, d\}$, then for a set $\mathbf{u}\subseteq \{1,\ldots d\}$
    $h_\pi(\mathbf{u}) = \arg\min_i \pi(i)$ for $i \in \mathbf{u}$. Then due to~\cite{Broder00,BroderCFM98,BroderCPM00}, 
    \begin{align*}\label{eq:cosine}
 \Pr[h_\pi(\mathbf{u})=h_\pi(\mathbf{v})]=\frac{|\mathbf{u}\cap \mathbf{v}|}{|\mathbf{u} \cup \mathbf{v}|}.
\end{align*}
\end{definition}

\subsection{Locality Sensitive Hashing}\label{subsection:subsecLSH}
LSH suggest an algorithm or alternatively a data structure for efficient approximate 
nearest neighbor ($c$-$\NN$) search in high dimensional space. 
We formally state it as follows:
 
 \begin{definition}{($c$-Approximate Nearest Neighbor $(c$-$\NN)).$}\label{definition:cNN}
 Let $\D$ be set  of points in $\R^d$, and $\Sim(.,.)$ be a desired similarity measure. Then for parameters $S, c>0$, 
 the $c$-$\NN$ problem is to construct a data structure that given any query point 
 $\mathbf{q} \in \D$   
 reports a $cS$-near neighbor of $\mathbf{q}$ in $\D$ if 
 there is an $S$-near neighbor of $\mathbf{q}$ in $D$.
Here, we say a point $\mathbf{x}\in \D$ is $S$-near neighbor of $\mathbf{q}$ if $\Sim(\mathbf{q}, \mathbf{x})>S.$
 \end{definition}
  
In the following we   define the concept of locality sensitive hashing (LSH) which suggest a
data structure to solve $c$-$\NN$ problem.
 \begin{definition}[Locality sensitive hashing~\cite{IM98}]\label{definition:LSH}
  Let $\D$ be a set of $n$ vectors in $\R^d$, and $U$ be the hashing
universe. Then, a family $\mathcal{H}$ of functions from $\D$ to $U$
is called as $(S, cS, p_1, p_2)$-sensitive for a similarity measure $\Sim(.,.)$ if for any $\mathbf{x}, \mathbf{y} \in D$,
\begin{itemize}
    \item if $\Sim(\mathbf{x}, \mathbf{y})\geq S$, then $\displaystyle \Pr_{h \in \mathcal{H}}[h(\mathbf{x})=h(\mathbf{y})]\geq p_1$,
    \item if $\Sim(\mathbf{x}, \mathbf{y})\leq cS$, then $\displaystyle \Pr_{h \in \mathcal{H}}[h(\mathbf{x})=h(\mathbf{y})]\leq p_2.$
\end{itemize}
\end{definition}
Clearly, any such scheme is interesting only when $p_1> p_2$, and $c<1$. Let $K, L$ be 
the parameters of the  data structure for LSH, where  $K$ is the number of hashes in each hash table, 
and $L$ is the number of hash tables, then due to \cite{IM98,GIM99}, we 
have $K=O\left(\log_{\frac{1}{p_2}} n\right)$ and 
$L=O\left(n^{\rho}\log n\right)$, where $\rho=\frac{\log p_1}{\log p_2}.$
Thus, given a family of $(S, cS, p_1, p_2)$-sensitive hash functions, and using result of~\cite{IM98,GIM99}, one can 
construct a data structure for $c$-$\NN$ with $O(n^{\rho}\log n)$ query time and space $O(n^{1+\rho}).$
\subsubsection{How to convert similarity preserving compression schemes to LSH ?}
LSH schemes for various similarity measures can be viewed as first compressing the 
input such that it preserve the desired similarity measure, and then applying collision 
based hashing on top of it. If any similarity preserving compression scheme provides a 
similar guarantee as of Definition~\ref{definition:LSH}, then for parameters -- similarity 
threshold $S$, and $c$, one can construct data 
structure for LSH (hash-tables with parameters $K$ and $L$) for the $c$-$\NN$ problem via \cite{IM98,GIM99}.

\section{A compression scheme for high dimensional sparse binary data}\label{sec:BinaryResult}
We first formally define our Compression Scheme as follows:
\begin{definition}[\textbf{B}inary \textbf{C}ompression \textbf{S}cheme]\label{defi:bcs}
 Let $\N$ be the number of buckets, for $i=1$ to $d$, we randomly assign 
 the $i$-th  position to a bucket number $b(i)$  $\in \{1, \ldots \N\}$. Then a vector $\mathbf{u} \in \{0, 1\}^d$,  
 compressed into a vector $\mathbf{u}' \in \{0, 1\}^{\N}$ as follows:
 \[\mathbf{u}'[j] = \sum_{i : b(i) = j} \mathbf{u}[i]  \pmod 2.\]  
\end{definition}

\begin{note}
 For brevity we denote the Binary Compression Scheme as $\BCS$. 
\end{note}
\paragraph{Some intuition} Consider two binary vectors $\mathbf{u}, \mathbf{v} \in\{0, 1\}^d$,  we call a bit position 
 \textit{``active''} if at least one of the vector between $\mathbf{u}$ and $\mathbf{v}$ has value $1$ in that  position. 
Let  $\tb$ be the maximum number of $1$ in any vector, then
  there could be at most $2\tb$ active positions shared between vectors $\mathbf{u}$ and $\mathbf{v}$. 
Further,  using the $\BCS$, let   $\mathbf{u}$ and $\mathbf{v}$ get 
compressed into binary vectors $\mathbf{u'}, \mathbf{v'} \in \{0, 1\}^{\N}$.
In the compressed vectors, we call a particular bit position \textit{``pure''} if the number of 
active positions mapped to that position is at 
most one, otherwise we call it \textit{``corrupted''}. It is easy to see that the contribution of pure bit 
positions in $\mathbf{u'}, \mathbf{v'}$ towards Hamming distance (or Inner product similarity), is exactly equal to the 
contribution of the bit positions in $\mathbf{u}, \mathbf{v}$ which get mapped to the pure bit positions. 
 \begin{figure}[ht!]
\centering
\includegraphics[scale=.032]{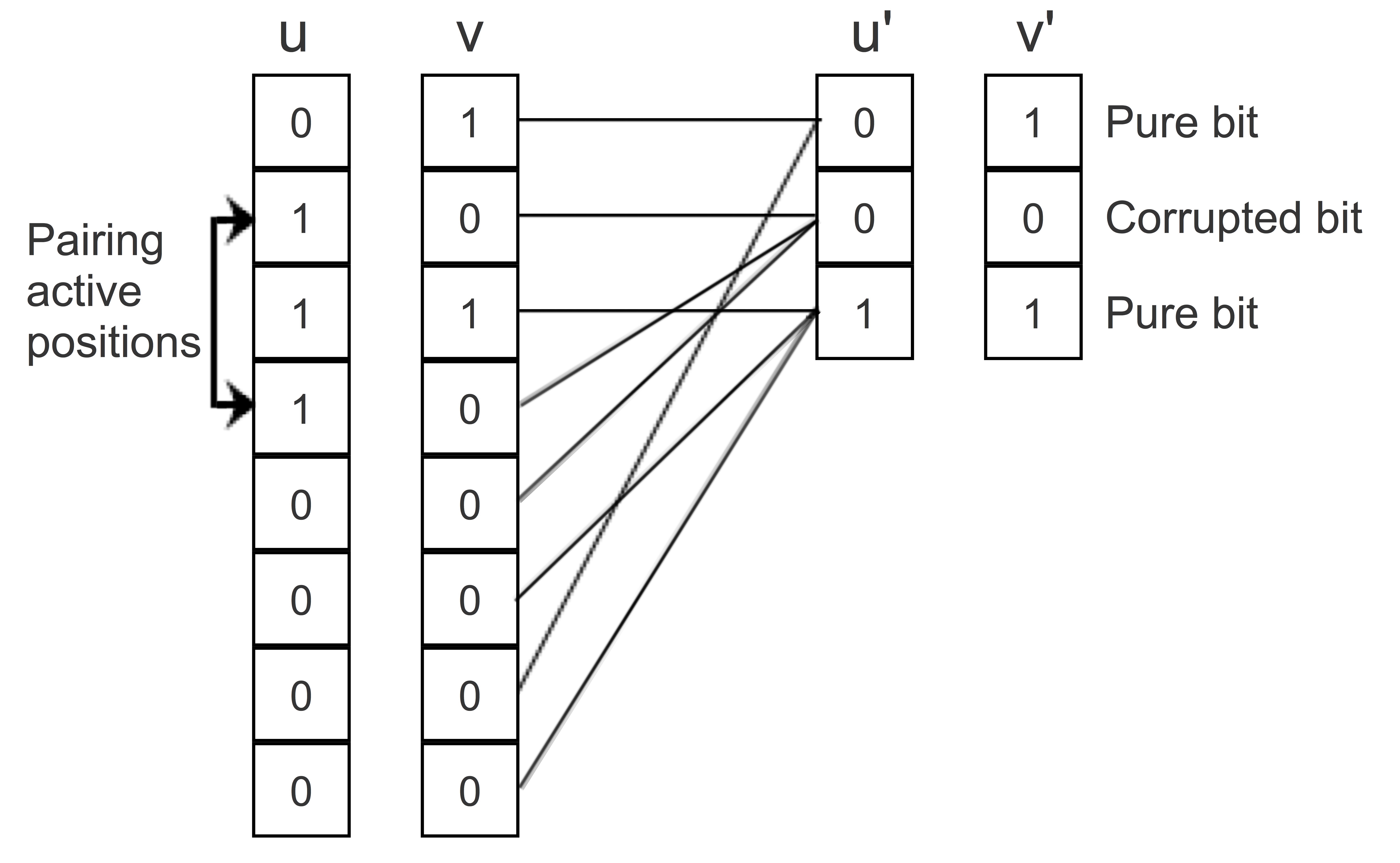}
\end{figure}
 The number of maximum possible corrupted bits in the compressed data is  $\tb$ because in the worst case 
it is possible that all the $2\tb$ active bit position got paired up while compression. 
The deviation of Hamming distance (or Inner product similarity) between $\mathbf{u'}$ and $\mathbf{v'}$ from that 
of $\mathbf{u}$ and $\mathbf{v}$, corresponds to the number of corrupted bit positions shared between 
$\mathbf{u'}$ and $\mathbf{v'}$. The above figure illustrate this with an example, and the lemma below analyse it.
 \begin{lem}\label{lem:compression}
 Consider two binary vectors $\mathbf{u}, \mathbf{v} \in\{0, 1\}^d$, which get compressed into  
 vectors $\mathbf{u'}, \mathbf{v'} \in \{0, 1\}^{\N}$ 
 using the $\BCS$, and  suppose $\tb$ is the maximum number of $1$ in any vector. 
 Then for an  integer $r\geq1$, and 
  $\epsilon>0$, probability that $\mathbf{u'}$ and $\mathbf{v'}$ share more than $\epsilon r$ 
  corrupted positions is at most
 $\left(\frac{2\tb}{\sqrt{\N}}\right)^{\epsilon r}.$
\end{lem}
\begin{proof}
 We first calculate the probability that a particular bit position gets corrupted between $\mathbf{u'}$ and $\mathbf{v'}$. 
 As there are at most $2\tb$ active positions shared between vectors $\mathbf{u}$ and $\mathbf{v}$, the number of ways of
 pairing two active positions from $2\tb$ active positions is  at most $2\tb \choose 2$, and this 
 pairing will result a corrupted bit position in $\mathbf{u'}$ or $\mathbf{v'}$. 
 Then, the  probability that a particular bit position in 
 $\mathbf{u'}$ or $\mathbf{v'}$ gets corrupted is at most $\frac{{2\tb \choose 2}}{{\N}}\leq \left(\frac{4{\tb}^2}{\N} \right).$
 Further, if the deviation of Hamming distance (or Inner product similarity) between $\mathbf{u'}$ and $\mathbf{v'}$ 
 from that of  $\mathbf{u}$ and $\mathbf{v}$ is more than $\epsilon r$, then at least $\epsilon r$ corrupted 
 positions are shared between $\mathbf{u'}$ and $\mathbf{v'}$, 
 which implies that at least $\frac{\epsilon r}{2}$ pair of active positions in $\mathbf{u}$ and $\mathbf{v}$
 got paired up while compression. 
 The number of possible ways of pairing $\frac{\epsilon r}{2}$ active positions from $2\tb$ active positions 
 is at most ${2\tb \choose \frac{\epsilon r}{2}}{2\tb-\frac{\epsilon r}{2} \choose \frac{\epsilon r}{2}} \frac{\epsilon r}{2}!\leq
 (2\tb)^{\epsilon r}.$ 
 Since the probability that a pair of active positions got mapped in the same bit position in the compressed data 
 is $\frac{1}{\N}$, the probability that  $\frac{\epsilon r}{2}$  pair of active positions got mapped 
 in $\frac{\epsilon r}{2}$ distinct bit positions in the compressed data is at 
 most $(\frac{1}{\N})^{\frac{\epsilon r}{2}}$. 
   Thus, by union bound, the probability that at least $\epsilon r$ corrupted bit 
 position shared between $\mathbf{u'}$ and $\mathbf{v'}$  is at most $\frac{(2\tb)^{\epsilon r}}{{\N}^{\frac{\epsilon r}{2}}}
 =\left(\frac{2\tb}{\sqrt{\N}}\right)^{\epsilon r}.$
 \end{proof}
 In the following lemma we generalize the above result on a set of $n$ binary vectors. 
 We suggest a compression bound such that   any pair of  compressed vectors share 
 only a very small number of corrupted bits, with high probability.  
 \begin{lem}\label{lem:compressionBound}
  Consider a set~$\mathrm{U}$ of $n$  binary vectors \\$\{\mathbf{u_i}\}_{i=1}^n\subseteq \{0, 1\}^d$, 
  which get compressed into  a set  $\mathrm{U'}$ of binary
  vectors $\{\mathbf{u_i'}\}_{i=1}^n\subseteq\{0, 1\}^{\N}$  using the $\BCS$. Then for any 
  positive  integer $r$, and   $\epsilon>0$, 
 \begin{itemize}
  \item if $\epsilon r >3 \log n$, and  we set $\N=16{\tb}^2$, then probability 
  that   for all $\mathbf{u_i'}, \mathbf{u_j'}\in \mathrm{U'}$
  share more than $\epsilon r$   corrupted positions is at most  $\frac{1}{n}$. 
  \item If $\epsilon r < 3 \log n$, and  we set $\N=144{\tb}^2\log^2n $, then probability 
  that   for all $\mathbf{u_i'}, \mathbf{u_j'}\in \mathrm{U'}$
  share more than $\epsilon r$   corrupted positions is at most  $\frac{1}{n}$.
 \end{itemize}
 \end{lem}
   \begin{proof}
 In the first case, for a fixed pair of compressed 
 vectors $\mathbf{u_i'}$ and $\mathbf{u_j'}$, due to lemma~\ref{lem:compression}, probability that they
 share more than $\epsilon r$ corrupted positions is at most  $\left(\frac{2{\tb}}{\sqrt{\N}}\right)^{\epsilon r}.$
 If $\epsilon r >3 \log n$, and  $\N=16{\tb}^2$, then the above probability is at most 
 $\left(\frac{2\tb}{\sqrt{\N}}\right)^{\epsilon r}<\left(\frac{2{\tb}}{4t}\right)^{3 \log n}=\left(\frac{1}{2}\right)^{3 \log n}<\frac{1}{n^3}.$ As there are at most 
 ${n \choose 2}$ pairs of vectors, then the probability of  every pair of compressed vectors  share more than $\epsilon r$ 
  corrupted positions is at most  $\frac{{n \choose 2}}{n^3}<\frac{1}{n}$. 
  
  In the second case, 
 as $\epsilon r <3 \log n$, we cannot upper bound the desired probability similar to the first case. 
 Here we use a trick, in the input data we replicate each bit position $3 \log n$ times, 
  which makes a $d$ dimensional 
 vector to a $3d\log n$  dimensional, and as a consequence the  Hamming distance 
 (or Inner product similarity) is also scaled up by a multiplicative factor of $3 \log n$.
We now apply the compression scheme on these scaled vectors, then for a fixed pair of compressed 
 vectors $\mathbf{u_i'}$ and $\mathbf{u_j'}$,  probability that they  have more than 
 $3 \epsilon r \log n $ corrupted positions is at most    
 $\left(\frac{6\tb\log n }{\sqrt{\N}}\right)^{3 \epsilon r\log n}$. As we set 
 $\N=144{\tb}^2\log^2n $, the above probability is at most 
 $\left(\frac{6\tb\log n }{\sqrt{144{\tb}^2\log^2n}}\right)^{3 \epsilon r \log n}< \left(\frac{1}{2}\right)^{3 \log n}<\frac{1}{n^3}.$
 The final probability follows by applying union bound over all ${n \choose 2}$ pairs.
  \end{proof}
  \begin{rem}
   We would like to emphasize that using the $\BCS$, for any pair of vectors, 
   the Hamming distance between them in the compressed version is always less than or equal to 
   their original Hamming distance. Thus, this compression scheme has only one-sided-error for 
   the Hamming case. However, in the case of inner product similarity this compression 
   scheme  can possibly have two-sided-error -- 
   as the inner product in the compressed version  can be    smaller or higher than the inner product of original input. 
   We illustrate this by the following example, 
   where the compression scheme assigns both bit positions of the input to one bit of the compressed data. 
  \begin{itemize}
 \item If $\mathbf{u}=[1, 0]~ \mbox{and~} \mathbf{v}=[0, 1]$, then $\IP(\mathbf{u}, \mathbf{v})=0$; 
 and  after compression $\mathbf{u'}=[1]\mbox{~and~} \mathbf{v'}=[1]$ which 
gives $\IP(\mathbf{u'}, \mathbf{v'})=1$. 
\item If  $\mathbf{u}=[1, 1]~ \mbox{and~} \mathbf{v}=[1, 1]$, then $\IP(\mathbf{u}, \mathbf{v})=2$, 
and after compression $\mathbf{u'}=[0]\mbox{~and~} \mathbf{v'}=[0]$ which 
gives $\IP(\mathbf{u'}, \mathbf{v'})=0.$
\end{itemize}
\end{rem}

As a consequence of Lemma~\ref{lem:compressionBound} and the above remark, we present our compression 
guarantee for the Hamming distance and Inner product similarity.
{ \renewcommand{\thetheorem}{\ref{theorem:compressionHamming}}
 \begin{theorem}
  Consider a set $\mathrm{U}$ of binary vectors \\ $\{\mathbf{u_i}\}_{i=1}^n\subseteq \{0, 1\}^d$, 
  a positive  integer $r$, and   $\epsilon>0$. 
  If $\epsilon r >3 \log n$,  we set $\N=O({\tb}^2)$;  if $\epsilon r < 3 \log n$,  
    we set $\N=O({\tb}^2\log^2n) $,  and compress them 
  into  a set  $\mathrm{U'}$ of binary vectors $\{\mathbf{u_i'}\}_{i=1}^n\subseteq\{0, 1\}^{\N}$  using 
  $\BCS$.  Then for all  $\mathbf{u_i}, \mathbf{u_j}\in \mathrm{U}$, 
\begin{itemize}
 \item if $\dH(\mathbf{u_i}, \mathbf{u_j})< r$, then $\Pr [\dH({\mathbf{u_i}}', {\mathbf{u_j}}')< r]=1$,
 \item if $\dH(\mathbf{u_i}, \mathbf{u_j})\geq (1+\epsilon)r$, then $\Pr [\dH({\mathbf{u_i}}', {\mathbf{u_j}}')< r]<\frac{1}{n}.$
\end{itemize}
\end{theorem}\addtocounter{theorem}{-1}}
{\renewcommand{\thetheorem}{\ref{theorem:compressionIP}}
 \begin{theorem}
 Consider a set $\mathrm{U}$ of binary vectors \\$\{\mathbf{u_i}\}_{i=1}^n\subseteq \{0, 1\}^d$, 
 a positive  integer $r$, and   $\epsilon>0$. 
  If $\epsilon r >3 \log n$,  we set $\N=O({\tb}^2)$;  if $\epsilon r < 3 \log n$,  
    we set $\N=O({\tb}^2\log^2n) $,  and compress them into 
    a set  $\mathrm{U'}$ of binary vectors
    $\{\mathbf{u_i'}\}_{i=1}^n\subseteq\{0, 1\}^{\N}$  using 
  $\BCS$.  Then for all $\mathbf{u_i}, \mathbf{u_j}\in \mathrm{U}$ the following is true with probability
 at least $1-\frac{1}{n},$
 \[
  (1-\epsilon)\IP(\mathbf{u_i}, \mathbf{u_j})\leq \IP({\mathbf{u_i}}', {\mathbf{u_j}}')\leq (1+\epsilon)\IP(\mathbf{u_i}, \mathbf{u_j}).
 \]
  \end{theorem}\addtocounter{theorem}{-1}}

\subsection{A tighter analysis for Hamming distance} 

In this subsection, we strengthen our analysis for the Hamming case, 
and shows a compression bound which is independent of the  dimension 
and the sparsity, and depends only on the Hamming distance between 
the vectors. 
However, we could show our result in expectation, and only for a pair of vectors.

For a pair of vectors $\mathbf{u}, \mathbf{v}\in \{0, 1\}^d$, we say 
that a bit position is \textit{``unmatched''} if exactly one of the vector has value $1$ in that position 
and the other one has value $0$.  We say that a bit position in the compressed data is \textit{``odd-bit''} if odd 
number of unmatched positions get mapped to that bit. Let $\mathbf{u}$ and $\mathbf{v}$ get compressed into vectors $\mathbf{u'}$ and $\mathbf{v'}$ using the 
$\BCS$. Our observation is that 
each odd bit position in the compressed data contributes to Hamming distance in $1$ in the compressed data. 
We illustrate this with an example: let $\mathbf{u}[i,j,k]=[1,0,1],~\mathbf{v}[i,j,k]=[0,1,0]$
and let $i, j, k$ get mapped to bit position $i'$ (say) in the compressed data, then $\mathbf{u}[i']=0, \mathbf{v}[i']=1$,
 then clearly $\dH(\mathbf{u}[i'], \mathbf{v}[i'])=1.$
{\renewcommand{\thetheorem}{\ref{theorem:compressionR}}
\begin{theorem}
 Consider two binary vectors $\mathbf{\mathbf{u}}, \mathbf{v} \\\in\{0, 1\}^d$, which get compressed into  
 vectors $\mathbf{\mathbf{u'}}, \mathbf{v'} \in \{0, 1\}^{\N}$  using $\BCS$. If we set  $\N=O(r^2)$, then
 \begin{itemize}
  \item if $\dH(\mathbf{u}, \mathbf{v})< r$, then $\Pr [\dH({\mathbf{u}}', {\mathbf{v}}')< r]=1$, and 
  \item if $\dH(\mathbf{\mathbf{u}}, \mathbf{v})\geq 4r$,  then $\E[\dH(\mathbf{\mathbf{u'}}, \mathbf{v'})]>2r.$
 \end{itemize}
 \end{theorem}\addtocounter{theorem}{-1}}
\begin{proof}
Let  $\tu$  denote the number of unmatched bit positions between $\mathbf{u}$ and $\mathbf{v}$.  
As mentioned earlier, if odd number of unmatched bit positions gets mapped to a 
particular bit in the compressed data, then that bit position corresponds 
to the Hamming distance $1$. Let we call that bit position as \textit{``odd-bit''} 
position. In order to give a bound on the Hamming distance in the compressed data 
we need to give a bound on number of such odd-bit positions. We first calculate the probability 
that a particular bit position say $k$-th position in the compressed data is odd. 
Let we denote this by ${\Pr}_{odd}^{(k)}$.
We  do it using the following binomial distribution:
\begin{align*}
 {\Pr}_{\mathrm{odd}}^{\mathrm{(k)}}&=\sum_{{i\bmod 2}=1}^{\tu} \frac{1}{\N^i}{\tu \choose i}\left(1-\frac{1}{\N}\right)^{\tu-1}.
\end{align*}
Similarly,  we compute the probability that the $k$-th bit is even:
\[
 {\Pr}_{\mathrm{even}}^{\mathrm{(k)}}=\sum_{{i\bmod 2}=0}^{\tu} \frac{1}{\N^i}{\tu \choose i}\left(1-\frac{1}{\N}\right)^{\tu-1}.
\]
We have,
\[
 {\Pr}_{\mathrm{even}}^{\mathrm{(k)}}+{\Pr}_{\mathrm{odd}}^{\mathrm{(k)}}=1. \numberthis\label{eq:eq1}
\]
Further,
\begin{align*}
 &{{\Pr}_{\mathrm{even}}^{\mathrm{(k)}}-{\Pr}_{\mathrm{odd}}^{\mathrm{(k)}}}\\
 &=\sum_{{i\bmod 2}=0}^{\tu} \frac{1}{\N^i}{\tu \choose i}\left(1-\frac{1}{\N}\right)^{\tu-1}\\&-\sum_{{i\bmod 2}=1}^{\tu} \frac{1}{\N^i}{\tu \choose i}\left(1-\frac{1}{\N}\right)^{\tu-1}\\
 &=\left(1-\frac{1}{\N} -\frac{1}{\N} \right)^{\tu}\\
 &=\left(1-\frac{2}{\N}\right)^{\tu}.\numberthis\label{eq:eq2}
\end{align*}
Thus, we have the following from Equation~\ref{eq:eq1} and Equation~\ref{eq:eq2} 
\begin{align*}
 {\Pr}_{\mathrm{odd}}^{\mathrm{(k)}}&=\frac{1}{2}\left(1-  \left(1-\frac{2}{\N}\right)^{\tu}\right)\\
 &\geq \frac{1}{2}\left(1- \exp\left(-\frac{2\tu}{\N}\right)\right).  \numberthis\label{eq:eq3}
\end{align*}
The last inequality follows as $(1-x)\leq e^x$ for $x<1.$
Thus expected number of odd-bits is at least $$\frac{\N}{2}\left(1- \exp\left(-\frac{2\tu}{\N}\right)\right).$$
We now split here in two cases: $1)$  $\tu<20r$, and $2)$  $\tu\geq 20r$. We address them one-by-one.

\textbf{Case 1:} $\tu< 20r$. 
We complete this case using Lemma~\ref{lem:compression}. It is easy to verify that in 
the case of Hamming distance the analysis of Lemma~\ref{lem:compression} also holds if 
we consider ``unmatched'' bits instead of ``active'' bits in the analysis. 
Thus, the probability that at least $r$ corrupted bit 
 position shared between $\mathbf{u'}$ and $\mathbf{v'}$  is at most $\left(\frac{2\tu}{\sqrt{\N}}\right)^{r}.$
We wish to set the value of $\N$ 
such that with probability at most $1/3$ that $\mathbf{u'}$ and $\mathbf{v'}$ share more than $r$ corrupted positions.
If we set the value of $\N=4{\tu^2}3^{\frac{2}{r}}$, then 
the above
probability is at most $\left({\frac{2\tu}{\sqrt{4{\tu^2}3^{\frac{2}{r}}}}}\right)^r=\frac{1}{3}.$
Thus, when $\N=4{\tu^2}3^{\frac{2}{r}}=O(\tu^2)=O(r^2)$ as $\tu<20r$ and $r\geq 2$, with probability at most $1/3$, 
at most $r$ corrupted bits 
are shared between $\mathbf{u'}$ and $\mathbf{v'}$. As a consequence to this, we have $\E[\dH(\mathbf{u'}, \mathbf{v'})]>\frac{2}{3}.3r=2r.$

\textbf{Case 2:} $\tu\geq 20r$. 
We continue here from Equation~\ref{eq:eq3}
\begin{align*}
 &\text{Expected number of odd buckets}\\&\geq \frac{\N}{2}\left(1- \exp\left(-\frac{2\tu}{\N}\right)\right)\\
 &\geq \frac{\N}{2}\left(1- \exp\left(-\frac{40r}{\N}\right)\right)\\
 &= 4r^2\left(1- \exp\left(-\frac{5}{r}\right)\right)   \numberthis\label{eq:eq4}\\ 
 &> 4r^2\left(\frac{1}{2r}\right)  \numberthis\label{eq:eq5}\\
 &=2r.
 \end{align*}
 Equality~\ref{eq:eq4} follows by setting $\N=8r^2$ and
 Inequality~\ref{eq:eq5} holds as  $1-\exp\left(-\frac{5}{r}\right)>\frac{1}{2r}$ for $r \geq 2$.
 
 Finally, Case $1$ and Case $2$ complete a proof of the theorem.
 \end{proof}

\section{A compression scheme for high dimensional sparse real data}\label{sec:RealResult}
We first define our compression scheme for the real valued data.
\begin{definition}(\textbf{R}eal-valued \textbf{C}ompression \textbf{S}cheme)\label{defi:rcs}
 Let $\N$ be the number of buckets, for $i=1$ to $d$, we randomly assign 
 the $i$-th  position to the bucket number $b(i)$  $\in \{1, \ldots \N\}$. 
 Then, for $j=1 \text{~to~} \N$, the $j$-th coordinate of the compressed vector 
$\boldsymbol{\alpha}$ is computed as follows:
 \[\boldsymbol{\alpha}[j] = \sum_{i : b(i) = j} \mathbf{a}[i]x_i,\]  
where each $x_i$ is a random variable that takes a value between $\{-1, +1\}$ with probability $1/2.$
\end{definition}
\begin{note}
 For brevity we denote our  Real-valued Compression Scheme as $\RCS$. 
\end{note}
We first present our compression guarantee  for preserving Inner product  for a pair of real valued vectors. 
\begin{lem}\label{lem:innerprod}
 Consider two  vectors $\mathbf{a}, \mathbf{b} \in\R^d$, which get compressed into  
 vectors $\boldsymbol{\alpha}, \boldsymbol{\beta} \in \R^{\N}$  using the 
 $\RCS$.  If we set $\N=\frac{10\tm^2}{\epsilon^2}$, 
 where $\tm=\max\{||\mathbf{a}||^2, ||\mathbf{b}||^2\}$ and $\epsilon>0$, then the following holds, 
 \[
 \Pr\left[\left|\langle\boldsymbol{\alpha}, \boldsymbol{\beta}\rangle- \langle\mathbf{a}, \mathbf{b}\rangle \right|>\epsilon \right]<1/10.
 \]
\end{lem}
\begin{proof}
 Let we have two vectors $\mathbf{a}, \mathbf{b}\in \R^d$ such that $\mathbf{a}=[a_1, a_2,\ldots a_d]$ and 
 $\mathbf{b}=[b_1, b_2,\ldots b_d]$. Let $\{x_i\}_{i=1}^d$ be a set of $d$ random variables such that each 
 $x_i$ takes a value between $\{-1, +1\}$ with probability $1/2$,  $z_i^{(k)}$ be a random variable that takes 
 the value $1$ if $i$-th dimension of the vector is mapped to the $k$-th bucket of the compressed vector and $0$ otherwise.
 Using the  compression scheme $\RCS$, let vectors $\mathbf{a}, \mathbf{b}$ get 
 compressed into vectors $\boldsymbol{\alpha}$ and  $\boldsymbol{\beta}$,  
 where $\boldsymbol{\alpha}=[\alpha_1, .. \alpha_k, .. \alpha_{\N}]$ such that  $\alpha_k=\Sigma_{i=1}^da_ix_iz_i^{(k)}$,
 and $\boldsymbol{\beta}=[\beta_1, .. \beta_k, .. \beta_{\N}]$ such that  $\beta_k=\Sigma_{i=1}^db_ix_iz_i^{(k)}.$
 We now compute the inner product of the compressed vectors $\langle\boldsymbol{\alpha}, \boldsymbol{\beta}\rangle$.
 \begin{align*}
  \langle\boldsymbol{\alpha}, \boldsymbol{\beta}\rangle&=\sum_{k=1}^{\N} \alpha_k \beta_k
  =\sum_{k=1}^{\N} \left( \Sigma_{i=1}^da_ix_iz_i^{(k)}  \right)\left( \Sigma_{i=1}^db_ix_iz_i^{(k)}  \right)\\
  &=\sum_{k=1}^{\N} \left( \Sigma_{i=1}^da_ib_ix_i^2{z_i^{(k)}}^2 + \Sigma_{i\neq j}a_ib_jx_ix_jz_i^{(k)} z_j^{(k)} \right)\numberthis\label{eq:eq10}\\
  &=\sum_{k=1}^{\N} \left( \Sigma_{i=1}^da_ib_i{z_i^{(k)}} + \Sigma_{i\neq j}a_ib_jx_ix_jz_i^{(k)} z_j^{(k)} \right)\numberthis\label{eq:eq11}\\
 &= \Sigma_{i=1}^da_ib_i\sum_{k=1}^{\N}{z_i^{(k)}} +\sum_{k=1}^{\N} \left(\Sigma_{i\neq j}a_ib_jx_ix_jz_i^{(k)} z_j^{(k)} \right)\\
  &=\Sigma_{i=1}^da_ib_i +\sum_{k=1}^{\N} \left(  \Sigma_{i\neq j}a_ib_jx_ix_jz_i^{(k)} z_j^{(k)} \right)\\
 &= \langle\mathbf{a}, \mathbf{b}\rangle +\sum_{k=1}^{\N} \left(  \Sigma_{i\neq j}a_ib_jx_ix_jz_i^{(k)} z_j^{(k)} \right).\numberthis\label{eq:eq12}
 \end{align*}
 Equation~\ref{eq:eq11} follows from Equation~\ref{eq:eq10} because  $x_i^2=1$ as $x_i=\pm1$, and $z_i^2=z_i$ as $z_i$ 
 takes value either $1$ or $0.$  We continue from Equation~\ref{eq:eq12} and 
 compute the Expectation and the Variance of the random variable 
 $\langle\boldsymbol{\alpha}, \boldsymbol{\beta}\rangle$. We first compute the Expectation of the 
 random variable $\langle\boldsymbol{\alpha}, \boldsymbol{\beta}\rangle$
 as follows:
 \begin{align*}
  \E[\langle\boldsymbol{\alpha}, \boldsymbol{\beta}\rangle]&=\E\left[\langle\mathbf{a}, \mathbf{b}\rangle +\sum_{k=1}^{\N} \left(  \Sigma_{i\neq j}a_ib_jx_ix_jz_i^{(k)} z_j^{(k)} \right)\right]\\
 &=\E [\langle\mathbf{a}, \mathbf{b}\rangle] +\E\left[\sum_{k=1}^{\N} \left(  \Sigma_{i\neq j}a_ib_jx_ix_jz_i^{(k)} z_j^{(k)} \right)   \right]\\
 &=\langle\mathbf{a}, \mathbf{b}\rangle +\sum_{k=1}^{\N} \left(  \Sigma_{i\neq j}\E[a_ib_jx_ix_jz_i^{(k)} z_j^{(k)}] \right) \numberthis\label{eq:eq13}  \\
 &=\langle\mathbf{a}, \mathbf{b}\rangle +\sum_{k=1}^{\N} \left(  \Sigma_{i\neq j}a_ib_j\E[x_ix_jz_i^{(k)} z_j^{(k)}] \right)   \\
 &=\langle\mathbf{a}, \mathbf{b}\rangle. \numberthis\label{eq:eq14}
 \end{align*}
 Equation~\ref{eq:eq13} holds due to the linearity of expectation. 
 Equation~\ref{eq:eq14} holds because $\E[x_ix_jz_i^{(k)} z_j^{(k)}]=0$ 
 as both $x_i$ and $x_j$ take a value between $\{-1, +1\}$ each with probability $0.5$ which leads to $\E[x_ix_j]=0$.
 
 We now compute the Variance  of the random variable $\langle\boldsymbol{\alpha}, \boldsymbol{\beta}\rangle$ as follows:
 \begin{align*}
  \Var[\langle\boldsymbol{\alpha}, \boldsymbol{\beta}\rangle]&=\Var\left[\langle\mathbf{a}, \mathbf{b}\rangle +\sum_{k=1}^{\N} \left(  \Sigma_{i\neq j}a_ib_jx_ix_jz_i^{(k)} z_j^{(k)} \right)\right]\\
 &=\Var\left[\sum_{k=1}^{\N} \left(  \Sigma_{i\neq j}a_ib_jx_ix_jz_i^{(k)} z_j^{(k)} \right)\right]\numberthis\label{eq:eq14}\\
 &=\Var\left[\sum_{k=1}^{\N} \left(  \Sigma_{i\neq j}\xi_{ij}^{(k)} \right)\right]\numberthis\label{eq:eq17}\\
 &=\Var\left[\sum_{i\neq j} \sum_{k=1}^{\N}\xi_{ij}^{(k)} \right]\\
 &=\sum_{i\neq j}\Var\left[\sum_{k=1}^{\N}\xi_{ij}^{(k)} \right]\ldots \\&+\sum_{i\neq j, i'\neq j', i\neq i', j\neq j'}\Cov\left[\sum_{k=1}^{\N}\xi_{ij}^{(k)}, \sum_{k=1}^{\N}\xi_{i'j'}^{(k)} \right]\numberthis\label{eq:eq16}
 \end{align*}
 Equation~\ref{eq:eq14} holds due to Fact~\ref{fact:varProp}; Equation~\ref{eq:eq17} holds as 
 we denote the expression $a_ib_jx_ix_jz_i^{(k)} z_j^{(k)}$ by the variable $\xi_{ij}^{(k)}$; 
 Equation~\ref{eq:eq16} holds due to Fact~\ref{fact:varProp1}. We now bound the values of the two terms 
 of Equation~\ref{eq:eq16}.
 \begin{align*}
  \sum_{i\neq j}\Var\left[\sum_{k=1}^{\N}\xi_{ij}^{(k)} \right]&=\sum_{i\neq j}\sum_{k=1}^{\N}\Var\left[\xi_{ij}^{(k)} \right]\ldots \\&~~~~+\sum_{i\neq j}\sum_{k\neq l}\Cov\left[\xi_{ij}^{(k)}, \xi_{ij}^{(l)} \right]\numberthis\label{eq:eq18}
 \end{align*}
 Equation~\ref{eq:eq18} holds due to Fact~\ref{fact:varProp1}. 
 We  bound the values of two terms 
 of Equation~\ref{eq:eq18} one by one as follows.
 \begin{align*}
 & \sum_{i\neq j}\sum_{k=1}^{\N}\Var\left[\xi_{ij}^{(k)} \right]=\sum_{i\neq j}\sum_{k=1}^{\N}\Var\left[a_ib_jx_ix_jz_i^{(k)} z_j^{(k)}\right]\\
  &=\sum_{i\neq j}a_i^2b_j^2\sum_{k=1}^{\N}\Var\left[x_ix_jz_i^{(k)} z_j^{(k)}\right]\numberthis\label{eq:eq19}\\
  &=\sum_{i\neq j}a_i^2b_j^2\sum_{k=1}^{\N}  \left( \E[x_i^2x_j^2{z_i^{(k)}}^2 {z_j^{(k)}}^2]-\E[x_ix_jz_i^{(k)} z_j^{(k)}]^2\right)\numberthis\label{eq:eq20}\\
 &=\sum_{i\neq j}a_i^2b_j^2\sum_{k=1}^{\N}   \E\left[{z_i^{(k)}} {z_j^{(k)}}\right]\numberthis\label{eq:eq21}\\
 &=\sum_{i\neq j}a_i^2b_j^2/\N\leq ||\textbf{a}||^2||\textbf{b}||^2/\N.\numberthis\label{eq:eq22}
 \end{align*}
 Equation~\ref{eq:eq19} holds due to Fact~\ref{fact:varProp}; Equation~\ref{eq:eq20} holds due to Definition~\ref{definition:varDef};
 Equation~\ref{eq:eq21} holds as $x_i^2, x_j^2=1$,  ${z_i^{(k)}}^2={z_i^{(k)}}$, and $\E[x_ix_j]=0$; finally, 
 Equation~\ref{eq:eq22} holds as $\sum_{i\neq j}a_i^2b_j^2\leq \sum_{i}a_i^2\sum_{i}b_i^2=||\textbf{a}||^2||\textbf{b}||^2.$ 

 We now bound the second term of Equation~\ref{eq:eq18}.
 \begin{align*}
  &\Cov\left[\xi_{ij}^{(k)}, \xi_{ij}^{(l)} \right]\\&=\Cov\left[a_ib_jx_ix_jz_i^{(k)} z_j^{(k)}, a_ib_jx_ix_jz_i^{(l)} z_j^{(l)} \right]\\
  &={a_i}^2{b_j}^2\Cov\left[x_ix_jz_i^{(k)} z_j^{(k)}, x_ix_jz_i^{(l)} z_j^{(l)} \right]\numberthis\label{eq:eq23}\\
  &={a_i}^2{b_j}^2\E[(x_ix_jz_i^{(k)} z_j^{(k)}-\E(x_ix_jz_i^{(k)} z_j^{(k)}))\\&~~~~~~~~~~~~~~(x_ix_jz_i^{(l)} z_j^{(l)}-\E(x_ix_jz_i^{(l)} z_j^{(l)})) ]\numberthis\label{eq:eq24}\\
 &={a_i}^2{b_j}^2\E\left[{x_i}^2{x_j}^2 z_i^{(k)} z_j^{(k)} z_i^{(l)} z_j^{(l)} \right]\numberthis\label{eq:eq25}\\
 &={a_i}^2{b_j}^2\E\left[z_i^{(k)} z_j^{(k)} z_i^{(l)} z_j^{(l)}\right]=0\numberthis\label{eq:eq26}
 \end{align*}
 Equation~\ref{eq:eq23} holds due to Fact~\ref{fact:coVarProp}; Equation~\ref{eq:eq24} holds 
 due to Definition~\ref{definition:coVarDef}; Equation~\ref{eq:eq25} holds as $\E(x_ix_j)=0$; finally, 
 Equation~\ref{eq:eq26} holds as in our compression scheme each 
 dimension of the input is mapped to a unique coordinate (bucket) in the compressed vector
 which implies that at least one of  the random variable  between $z_i^{(k)}$ and  $z_j^{(k)}$ has to be zero.
 
 We now bound the second term of Equation~\ref{eq:eq16}.
 \begin{align*}
 & \Cov\left[\sum_{k=1}^{\N}\xi_{ij}^{(k)}, \sum_{k=1}^{\N}\xi_{i'j'}^{(k)} \right]\\
  &=\E\left[\left(\sum_{k=1}^{\N}\xi_{ij}^{(k)}-\E(\sum_{k=1}^{\N}\xi_{ij}^{(k)}) \right)  \left(\sum_{k=1}^{\N}\xi_{i'j'}^{(k)}-\E(\sum_{k=1}^{\N}\xi_{i'j'}^{(k)}) \right) \right]\\
 &=\E\left[(\sum_{k=1}^{\N}\xi_{ij}^{(k)})(\sum_{k=1}^{\N}\xi_{i'j'}^{(k)}) \right]\numberthis\label{eq:eq27}\\
  &=\E\left[\left(\sum_{k=1}^{\N}a_ib_jx_ix_j{z_i}^{(k)}{z_j}^{(k)}\right)\left(\sum_{k=1}^{\N}a_{i'}b_{j'}x_{i'}x_{j'}{z_{i'}}^{(k)} {z_{j'}}^{(k)}\right)\right]\\
  &=a_ib_ja_{i'}b_{j'}\E \left[x_ix_jx_{i'}x_{j'}\left(\sum_{k=1}^{\N}{z_i}^{(k)}{z_j}^{(k)}\right)\left(\sum_{k=1}^{\N}{z_{i'}}^{(k)} {z_{j'}}^{(k)}\right)\right]\\
 &=0 \numberthis\label{eq:eq28}
 \end{align*}
 Equation~\ref{eq:eq27} holds as $\E(\sum_{k=1}^{\N}\xi_{ij}^{(k)})$ and $\E(\sum_{k=1}^{\N}\xi_{i'j'}^{(k)})$ is equal to zero because 
$$\E(\sum_{k=1}^{\N}\xi_{ij}^{(k)})=\sum_{k=1}^{\N}\E(\xi_{ij}^{(k)})=\sum_{k=1}^{\N}\E(a_ib_jx_ix_jz_i^{(k)} z_j^{(k)})=0.$$ 
 A similar argument follows for the other term as well. Equation~\ref{eq:eq28} holds as $\E[x_ix_jx_{i'}x_{j'}]$ is equal
to zero because each variable in the expectation term takes a value between $+1$ and $-1$ with probability $0.5.$

 Thus, we have 
$$\E[\langle\boldsymbol{\alpha}, \boldsymbol{\beta}\rangle]=\langle\mathbf{a}, \mathbf{b}\rangle,$$ and 
Equation~\ref{eq:eq16} in conjunction with Equations~\ref{eq:eq18}, \ref{eq:eq22}, \ref{eq:eq26}, \ref{eq:eq28} gives
$$\Var[\langle\boldsymbol{\alpha}, \boldsymbol{\beta}\rangle]\leq ||\textbf{a}||^2||\textbf{b}||^2/\N\leq {\tm}^2/{\N},$$
where $\tm=\max\{||\mathbf{a}||^2, ||\mathbf{b}||^2\}$.

Thus, by Chebyshev's inequality (see Fact~\ref{fact:Chebyshev}), we have 
\[
 \Pr\left[\left|\langle\boldsymbol{\alpha}, \boldsymbol{\beta}\rangle- \langle\mathbf{a}, \mathbf{b}\rangle \right|>\epsilon \right]<\frac{\tm^2}{\epsilon^2\N}=1/10.
 \]
 The last inequality follows as  we set $\N=\frac{10\tm^2}{\epsilon^2}$.
\end{proof}

Using a similar analysis we can generalize our result for $k$-way inner product. 
We state our result as follows:
{\renewcommand{\thetheorem}{\ref{theorem:compressionRealKway}}
\begin{theorem}
 Consider a set of $k$ vectors $\{\mathbf{a_i}\}_{i=1}^k\in \R^d$, 
  which get compressed into vectors $\{\boldsymbol{\alpha_i}\}_{i=1}^k \in \R^{\N}$   using the $\RCS$. 
   If we set $\N=\frac{10\tm^k}{\epsilon^2}$, 
 where $\tm=\max\{||{\mathbf{a_i}||^2}\}_{i=1}^k$ and $\epsilon>0$, then the following holds 
 \[
 \Pr\left[\left|\langle\boldsymbol{\alpha}_1\boldsymbol{\alpha}_2\ldots\boldsymbol{\alpha}_k\rangle- \langle\mathbf{a}_1\mathbf{a}_2\ldots\mathbf{a}_k\rangle \right|>\epsilon \right]<1/10.
 \]
\end{theorem}\addtocounter{theorem}{-1}}
We can also generalize the result of Lemma~\ref{lem:innerprod} for Euclidean distance as well. 
Consider a pair of vectors $\mathbf{a}, \mathbf{b}\in \R^d$ which get compressed into vectors 
$\boldsymbol{\alpha}, \boldsymbol{\beta}\in  \R^{\N}$ using the  compression scheme
$\RCS$. 
 Let $||\boldsymbol{\alpha}, \boldsymbol{\beta}||^2$  denote the squared euclidean distance between 
 the vectors $\boldsymbol{\alpha}, \boldsymbol{\beta}$. 
 Using a similar analysis of 
 Lemma~\ref{lem:innerprod} we can compute Expectation and Variance of the random variable $||\boldsymbol{\alpha}, \boldsymbol{\beta}||^2$
 \[
 \E[||\boldsymbol{\alpha}, \boldsymbol{\beta}||^2]=||\mathbf{a}, \mathbf{b}||^2, 
 \]
and 
\[
 \Var[||\boldsymbol{\alpha}, \boldsymbol{\beta}||^2]\leq \frac{(||a||^2-||b||^2)^2}{\N}\leq \frac{\tm^2}{\N}, 
\]
where $\tm=\max\{||\mathbf{a}||^2, ||\mathbf{b}||^2\}$. 
Thus, due to Chebyshev's inequality (see Fact~\ref{fact:Chebyshev}), we have the following result for Euclidean distance.
\begin{theorem}\label{theorem:compressionEuclidean}
Consider two  vectors $\mathbf{a}, \mathbf{b} \in\R^d$, which get compressed into  
 vectors $\boldsymbol{\alpha}, \boldsymbol{\beta} \in \R^{\N}$  using the 
 $\RCS$.  If we set $\N=\frac{10\tm^2}{\epsilon^2}$, 
 where $\tm=\max\{||\mathbf{a}||^2, ||\mathbf{b}||^2\}$ and $\epsilon>0$, then the following holds 
 \[
 \Pr\left[\left|||\boldsymbol{\alpha}, \boldsymbol{\beta}||^2- ||\mathbf{a}, \mathbf{b}||^2 \right|>\epsilon \right]<1/10.
 \]
 \end{theorem}
\begin{rem}
In order to compress a pair of data points our scheme requires $O(d\log \N)$ randomness,  
which grows logarithmically in the compression length, whereas  the 
other schemes require  randomness which  grows linearly in the compression length.
Thus, when the number 
of points are small (constant),  
then for preserving a pairwise Inner product or Euclidean distance, 
we have a clear advantage on  the amount of  
randomness required for the compression. We also believe that  using a more sophisticated  
concentration result (such as Martingale) it is possible to obtain a more tighter
concentration guarantee, and as a consequence a smaller compression length.  
\end{rem}

\section{Conclusion and open questions}\label{sec:Conclusion}
In this work, to the best of our knowledge, we obtain the first efficient binary to binary 
 compression scheme for preserving Hamming distance and Inner Product for high dimensional sparse data. For 
 Hamming distance in fact our scheme obtains the ``no-false-negative''
 guarantee analogous to the one obtained in recent paper by Pagh~\cite{Pagh16}.
 Contrary to the ``local'' projection approach of previous schemes we  first randomly 
 partition the dimension, and then take a ``global summary'' within a partition.
 The compression length of our scheme depends only on the sparsity and is independent 
 of the dimension as opposed to previously known schemes.
 We also obtain a generalization of our result to real-valued setting. Our work 
 leaves the possibility of several open questions -- improving the  bounds of our compression 
 scheme, and extending it to other similarity measures such as Cosine and Jaccard similarity 
 are major open questions of our work.

\bibliographystyle{abbrv}
\bibliography{reference}  
\end{document}